\documentclass[12pt]{article}
\usepackage{amsmath,amssymb}
\usepackage{times}
\usepackage{array}
\usepackage[english]{babel}

\usepackage{amssymb,amsthm}

\newtheorem{defn}{Definition}
\newtheorem{lemma}{Lemma}
\newtheorem{prop}{Proposition}

\newtheorem{thm}{Theorem}

\theoremstyle{definition}

\title{\bf Sobolev spaces on locally compact abelian groups}
\author{ Przemys{\l}aw G\'{o}rka$^1$ and Enrique G. Reyes$^2$ \\
\small{$^1$Instituto de Matem\'atica y F\'isica, Universidad de Talca}\\
\small{ Casilla 747, Talca, Chile, and}\\
\small{Department of Mathematics and Information Sciences,}\\
\small{Warsaw University of Technology,}\\
\small{Pl. Politechniki 1, 00-661 Warsaw, Poland.} \\
\small{$^{2}$ Departamento de Matem\'atica y Ciencia de la
Computaci\'on,} \\
\small{ Universidad de Santiago de Chile }\\
\small{Casilla 307 Correo 2, Santiago, Chile. }}

\begin{document}

\maketitle
\begin{abstract}
Motivated by a class of nonlinear equations of interest for string theory,
we introduce Sobolev spaces on arbitrary locally compact abelian groups
and we examine some of their properties. Specifically, we focus on analogs of the Sobolev
embedding and Rellich-Kondrachov compactness theorems. As an application, we prove the existence
of continuous solutions to a generalized euclidean bosonic string equation possed on an arbitrary
compact abelian group.
\end{abstract}


\date
\section{Introduction}

In this work we introduce Sobolev spaces on arbitrary locally compact abelian groups
and we examine analogs to the Sobolev embedding and Rellich-Kondrachov
compactness theorems. Sobolev spaces are well understood on (domains of) $\mathbb{R}^n$, see \cite{adams},
compact and complete Riemannian manifolds \cite{aubin,hebey}, and metric measure
spaces (the so-called Hajlasz-Sobolev spaces, see \cite{H1,H2}, and Newtonian spaces \cite{S}).
There are also some works on Sobolev spaces in the $p$-adic context, see \cite{RV} and references
therein, and in special cases of locally compact groups such as the Heisenberg group \cite{BA}.

Besides their intrinsic interest, we are interested in Sobolev spaces in this general context because we wish to consider some nonlinear equations
appearing in physical theories \cite{BK1,GPR_CAOT,GPR_Lorentz,V,W} in
settings beyond Riemannian manifolds, and also because we wish to use them as a tool for a better understanding
of pseudodifferential operators defined on locally compact abelian groups, motivated by the papers \cite{RV}
and \cite{GS}, and also by the recent treatise \cite{RT} in which the authors study in detail pseudo-differential
operators on compact Lie groups. As a first application of our Sobolev-type theorems, we
investigate the existence of continuous solutions to the {\em generalized  euclidean bosonic string equation}
\begin{equation} \label{abs}
\Delta e^{-c\Delta}\phi = U(x,\phi)\; , \; \; \; c > 0 \; ,
\end{equation}
introduced in \cite{CMN} (in a Lorentzian context) and recently considered in \cite{GPR_CAOT,GPR_JDE}.
We stress that studying such an equation in the general setting of topological groups is not simply a
technical exercise. Equation (\ref{abs}) is obtained formally in \cite{CMN} as the Euler-Lagrange equation
of the ``nonlocal Lagrangian''
\begin{eqnarray}\label{jeden}
    \mathcal{L}(\phi) = \phi\, \Delta e^{-c\, \Delta}\phi - \mathcal{U}(x,\phi) \; , \; \; \; c > 0\; ,
\end{eqnarray}
and this Lagrangian is but an approximation to the highly sophisticated bosonic string action considered
in \cite{W} which contains an infinite number of fields and yields ---via a formal application of the
variational principle--- an infinite number of equations for infinitely many variables, see for instance
\cite{Ra}. Topological groups appear therefore as a natural testing ground for gathering a better
understanding of (\ref{abs}) and (\ref{jeden}). For instance, we can consider Equation (\ref{abs}) for
functions $\phi$ ``depending on an infinite number of variables'', if we pose it on an infinite product of
spheres.

We start with some standard notation from harmonic analysis \cite{hr1,hr2}.
Let us fix a locally compact abelian group $G$. We denote by $\mu_G$ the unique Haar measure of $G$.
We also consider the {\em dual group} of the group $G$ (that is, the locally compact abelian group of all
continuous group homomorphisms from $G$ to the circle group $T$), and we
denote it by $G^{\wedge}$. $L^p$ spaces over $G$ are defined as usual,
\[
 L^p_{\mu_G}(G) = \left\{ f : G \rightarrow {\mathbb C} \; : \; \int_G |f(x)|^p d\mu_G(x) < \infty \right\} \; ,
\]
and the Fourier transform on $G$ is defined as follows: if $f \in L^1_{\mu_G}(G)$, then it Fourier
transform is the function $\hat{f} : G^\wedge \rightarrow \mathbb{C}$ given by
\begin{eqnarray*}
    \hat{f}(\xi) = \int_G \overline{\xi(x)} f(x) d\mu_G(x) \; .
\end{eqnarray*}
Next, we denote by $\Gamma$ the following set
\begin{eqnarray*}
    \Gamma=
\left\{ \gamma : G^{\wedge} \rightarrow [0,\infty) :
\exists_{\,c_{\gamma}}\, \forall_{\alpha, \beta \in G^{\wedge}}\,\,
\gamma(\alpha \beta)\leq c_{\gamma}\left[\, \gamma(\alpha)+\gamma(\beta)\,\right]\right\} \; ,
\end{eqnarray*}
and we are in position to introduce Sobolev spaces:

\begin{defn} \label{sob}
    Let us fix a map $\gamma\in\Gamma$ and a nonnegative real number $s$. We shall say that
    $f\in L^2_{\mu_G}(G)$ belongs to $H^s_{\gamma}(G)$ if the following integral is finite:
     \begin{eqnarray} \label{int1}
    \int_{G^{\wedge}} \left(1+\gamma(\xi)^2\right)^s |\hat{f}(\xi)|^2 d\mu_{G^{\wedge}}(\xi) \; .
    \end{eqnarray}
    Moreover, for $f\in H^s_{\gamma}(G)$ its norm $\|f\|_{ H^s_{\gamma}(G)}$ is defined as follows:
    \begin{eqnarray} \label{int2}
    \|f\|_{ H^s_{\gamma}(G)} =
\left(\int_{G^{\wedge}} \left(1+\gamma(\xi)^2\right)^s |\hat{f}(\xi)|^2 d\mu_{G^{\wedge}}(\xi)\right)^{\frac{1}{2}}.
    \end{eqnarray}
\end{defn}

{\bf Remark:} We note that by taking appropriate functions $\gamma$ we obtain the classical Sobolev spaces
on $\mathbb{T}^n$ and $\mathbb{R}^n$, see \cite{FeiWe} and \cite[Chp. 4]{Tay}. The use of the Fourier
transform and the duality theory of locally compact abelian groups is crucial in the present general
context, since we do not have differential calculus at our disposal. A particular instance of Definition
\ref{sob} appears in the paper \cite{FeiWe} by H. G. Feichtinger and T. Werther. The function $\gamma$
used in that article is called by their authors a {\em weakly subadditive weight}. We also note that in
$p$-adic analysis Sobolev spaces are defined in a way analogous to our Definition \ref{sob}: if we take
$\gamma(\xi)=\|\xi\|_p$, where $\|.\|_p$ is a $p$-adic norm on
$\mathbb{Q}_p^n \simeq\mathbb{Q}_p^{n\wedge}$, then (\ref{int1}) and (\ref{int2}) allow us to
recover the $p$-adic Sobolev spaces considered in \cite{RV}.

\vspace{2mm}

{\bf Remark:} It is important to take $\gamma \in \Gamma$ in order to prove that ---as it happens in
standard contexts, see \cite{FeiWe,Tay}--- our spaces $ H^s_{\gamma}(G)$ are Banach algebras under some
assumptions on $s$. We show this fact in Theorem 2 below. The other theorems of Sections 2 and 3 of this
paper hold true without this assumption on $\gamma$.

\vspace{2mm}

{\bf Remark:} Our spaces $H^s_{\gamma}(G)$ are contained in the $A^p_{{\rm w}, \omega}(G)$ spaces
introduced by H.G. Feichtinger and A.T. G\"urkanli in \cite{FG} in the following sense: if (notation
as in \cite{FG}) ${\rm w} \in L^2_{\mu_G}(G)$  and we take $\omega = (1 + \gamma^2)^{s/2}$, then,
$H^s_{\gamma}(G) \hookrightarrow A^2_{{\rm w}, \omega}(G)$.

\section{Continuous embedding theorems}
Embedding properties of Sobolev spaces are essential for proving existence and regularity
of solutions to partial differential equations \cite{Tay} and for the analysis of pseudo-differential
operators, see for instance \cite{RV}.
Thus, we begin by proving a Sobolev embedding type theorem in our general setting.

Let us start with two elementary observations. First, we show in Proposition 1 that our spaces
$H^s_{\gamma}(G)$ are included in $L^2_{\mu_G}(G)$. Then, we prove in Proposition 2 that in fact we have a
``scale'' of spaces:
\begin{prop}
    If $G$ is a locally compact abelian group then,
     \begin{eqnarray*}
     H^s_{\gamma}(G) \hookrightarrow L^2_{\mu_G}(G) \; .
    \end{eqnarray*}
Moreover, for each $f \in H^s_{\gamma}(G)$ the following inequality holds:
     \begin{eqnarray*}
   \|f\|_{L^2_{\mu_G}(G)} \leq \|f\|_{ H^s_{\gamma}(G)} \; .
    \end{eqnarray*}
\end{prop}
\begin{proof}
By Pontriagin duality and a basic inequality we get
\begin{eqnarray*}
   \|f\|_{L^2_{\mu_G}(G)} = \|\hat{f}\|_{L^2(G^{\wedge})} =
\left(\int_{G^{\wedge}} |\hat{f}(\xi)|^2 d\mu_{G^{\wedge}}(\xi)\right)^{\frac{1}{2}}\leq  \\
\left(\int_{G^{\wedge}} \left(1+\gamma(\xi)^2\right)^s |\hat{f}(\xi)|^2 d\mu_{G^{\wedge}}(\xi)\right)^{\frac{1}{2}} =\|f\|_{ H^s_{\gamma}(G)} \; .
    \end{eqnarray*}
\end{proof}

\begin{prop}
If $s > \sigma$, then $H^s_{\gamma}(G) \hookrightarrow H^{\sigma}_{\gamma}(G)$. Moreover, the inequality
  \begin{eqnarray*}
  \|f\|_{ H^{\sigma}_{\gamma}(G)} \leq \|f\|_{ H^s_{\gamma}(G)}
    \end{eqnarray*}
  holds.
\end{prop}
\begin{proof}
The proof follows from an elementary inequality.
\end{proof}

The classical Sobolev embedding theorem, see for instance \cite{adams}, reads in our context as follows:

\begin{thm}
    If $\frac{1}{\left(1+\gamma(.)^2\right)^s} \in L^{1}(G^{\wedge})$, then
     \begin{eqnarray*}
     H^s_{\gamma}(G) \hookrightarrow C(G) \; ,
    \end{eqnarray*}
in which $C(G)$ denotes the space of continuous complex-valued functions on $G$.
Moreover, there exists a constant $C(\gamma,s)$  such that for each $f \in H^s_{\gamma}(G)$, the
following inequality holds:
     \begin{eqnarray*}
   \|f\|_{C(G)} \leq C(\gamma,s)\|f\|_{ H^s_{\gamma}(G)}.
    \end{eqnarray*}
\end{thm}
\begin{proof}
Using the formula for the inverse Fourier transform,
\[
f(x) = \int_{\widehat{G}} \widehat f(\xi)\xi(x)\;d\mu_{G^{\wedge}}(\xi) \; ,
\]
we get
\begin{eqnarray*}
    |f(x)| = \left|\int_{G^{\wedge}} \hat{f}(\xi) \xi(x) d\mu_{G^{\wedge}}(\xi)\right|\leq \int_{G^{\wedge}} \left|\hat{f}(\xi) \right| d\mu_{G^{\wedge}}(\xi)\\
    \leq \left(\int_{G^{\wedge}} \left(1+\gamma(\xi)^2\right)^s |\hat{f}(\xi)|^2 d\mu_{G^{\wedge}}(\xi)\right)^{\frac{1}{2}}
    \left\|\frac{1}{\left(1+\gamma(.)^2 \right)^s} \right\|^{\frac{1}{2}}_{L^1(G^{\wedge})}.
\end{eqnarray*}
Finally, since $\hat{f} \in L^1(G^{\wedge})$ we get that $f \in  C(G)$ (see \cite{hr2}, Theorem 31.5).
\end{proof}

The following theorem tells us that, under a technical assumption involving the exponent $s$ and our
function $\gamma$, the space $H^s_{\gamma}(G)$ is a Banach algebra. It is well-known that such a property
is important for instance, for the study of existence of solutions to partial differential equations.
A recent example appears in our paper \cite{GR}.

\begin{thm}
    If $\frac{1}{\left(1+\gamma(.)^2\right)^s} \in L^{1}(G^{\wedge})$, then $H^s_{\gamma}(G)$ is a Banach
    algebra. Moreover, there exists a constant $D(\gamma,s)$ such that for each $f,g \in H^s_{\gamma}(G)$,
    the following inequality holds
     \begin{eqnarray*}
   \|fg\|_{ H^s_{\gamma}(G)} \leq D(\gamma,s)\|f\|_{ H^s_{\gamma}(G)}\|g\|_{ H^s_{\gamma}(G)} \; .
    \end{eqnarray*}
\end{thm}
\begin{proof}
First of all let us notice that for each $\xi, \eta \in G^{\wedge}$ the following inequality holds
\begin{eqnarray*}
(1+\gamma(\xi)^2)\leq (2+2c_{\gamma}^2)(2 + \gamma(\xi \eta^{-1})^2+\gamma(\eta)^2) \; .
\end{eqnarray*}
Hence, we obtain
\begin{eqnarray*}
    (1+\gamma(\xi)^2)^{\frac{s}{2}} \widehat{fg}(\xi) =\int_{G^{\wedge}} (1+\gamma(\xi)^2)^{\frac{s}{2}} \widehat{f}(\xi\eta^{-1})\widehat{g}(\eta)\,d\mu_{G^{\wedge}}(\eta)\leq\\
    2^{\frac{s}{2}}(1+c_{\gamma}^2)^{\frac{s}{2}} \int_{G^{\wedge}} \left(1+\gamma(\xi\eta^{-1})^2 + 1+\gamma(\eta)^2\right)^{\frac{s}{2}} |\widehat{f}(\xi\eta^{-1})\widehat{g}(\eta)|\,d\mu_{G^{\wedge}}(\eta)\leq\\
    2^{s-1}(1+c_{\gamma}^2)^{\frac{s}{2}} \int_{G^{\wedge}} \left(1+\gamma(\xi\eta^{-1})^2 \right)^{\frac{s}{2}} |\widehat{f}(\xi\eta^{-1})\widehat{g}(\eta)|\,d\mu_{G^{\wedge}}(\eta)+\\
    2^{s-1}(1+c_{\gamma}^2)^{\frac{s}{2}} \int_{G^{\wedge}} \left(1+\gamma(\eta)^2\right)^{\frac{s}{2}} |\widehat{f}(\xi\eta^{-1})\widehat{g}(\eta)|\,d\mu_{G^{\wedge}}(\eta) =\\
    2^{s-1}(1+c_{\gamma}^2)^{\frac{s}{2}}\left((|(1+\gamma(.)^2)^{\frac{s}{2}} \widehat{f})\ast \widehat{g}| +
     |\widehat{f}\ast(\widehat{g}(1+\gamma(.)^2)^{\frac{s}{2}})|\right).
\end{eqnarray*}
Thus
\begin{eqnarray*}
   \|fg\|_{ H^s_{\gamma}(G)}^2=\|(1+\gamma(\xi)^2)^{\frac{s}{2}} \widehat{fg}(\xi)\|^2_{L^2(G^{\wedge})}\leq\\
     2^{2s-1}(1+c_{\gamma}^2)^{s}\left(\|(1+\gamma(.)^2)^{\frac{s}{2}} \widehat{f})\ast\widehat{g}\|^2_{L^2(G^{\wedge})} +
      \|\widehat{f}\ast(\widehat{g}(1+\gamma(.)^2)^{\frac{s}{2}})\|^2_{L^2(G^{\wedge})}\right).
\end{eqnarray*}
Next, by Young inequality ($\|u\ast v\|_{L^2(G^{\wedge})}\leq c_y
\|u\|_{L^2(G^{\wedge})}\|v\|_{L^1(G^{\wedge})}$), we obtain
\begin{eqnarray*}
   \|fg\|_{ H^s_{\gamma}(G)}^2\leq   2^{2s-1}(1+c_{\gamma}^2)^{s}c^2_y\left(\|f\|^2_{ H^s_{\gamma}(G)} \|\widehat{g}\|^2_{L^1(G^{\wedge})} +
      \|\widehat{f}\|^2_{L^1(G^{\wedge})}\|g \|^2_{H^s_{\gamma}(G)}\right).
\end{eqnarray*}
Finally, from the proof of the previous theorem we can finish the
proof.
\end{proof}

Now we prove a second embedding result. While Theorem 1 tells us that functions in $H^s_{\gamma}(G)$ are
continuous, Theorem 3 tells us that they possess ``higher integrability properties'':

\begin{thm}\label{em}
If $\alpha > s$ and $\frac{1}{\left(1+\gamma(.)^2\right)} \in L^{\alpha}(G^{\wedge})$, then
  \begin{eqnarray*}
     H^s_{\gamma}(G) \hookrightarrow L^{\alpha^{*}}(G) \; ,
  \end{eqnarray*}
where $\alpha^{*}= \frac{2\alpha}{\alpha - s}$. Moreover, there exists a constant $D(\gamma,s)$  such
that for each $f \in H^s_{\gamma}(G)$, the following inequality holds
     \begin{eqnarray*}
   \|f\|_{L^{\alpha^{*}}(G)} \leq D(\gamma,s)\|f\|_{ H^s_{\gamma}(G)} \; .
    \end{eqnarray*}
\end{thm}
\begin{proof}
By a standard corollary of Hausdorff-Young inequality (see \cite{hr2}) we have
 \begin{eqnarray*}
   \|f\|_{L^{\alpha^{*}}(G)} \leq \|\hat{f}\|_{ L^p(G^{\wedge})} \; ,
 \end{eqnarray*}
where $p$ is the conjugate of $\alpha^{*}$, i.e. $p = \frac{2\alpha}{\alpha + s}$.  Next, using H\"{o}lder
inequality with exponents $\frac{2}{p}$ and $\frac{2}{2-p}$ we get
\begin{eqnarray*}
   \|\hat{f}\|_{ L^p(G^{\wedge})}= \left(\int_{G^{\wedge}}|\hat{f}(\xi)|^p \frac{\left(1+\gamma(\xi)^2\right)^{\frac{sp}{2}}}{\left(1+\gamma(\xi)^2\right)^{\frac{sp}{2}}}  d\mu_{G^{\wedge}}(\xi)\right)^{\frac{1}{p}}\leq \\\leq \|f\|_{ H^s_{\gamma}(G)}\left(\int_{G^{\wedge}}\frac{1}{\left(1+\gamma(\xi)^2\right)^{\frac{sp}{2-p}}}  d\mu_{G^{\wedge}}(\xi)\right)^{\frac{2-p}{2p}}.
 \end{eqnarray*}
Since $\frac{sp}{2-p}=\alpha$, we get
\begin{eqnarray*}
\|f\|_{L^{\alpha^{*}}(G)} \leq \left\|\frac{1}{\left(1+\gamma(.)^2\right)}\right\|_{L^{\alpha}(G^{\wedge})}^{\frac{s}{2}} \|f\|_{ H^s_{\gamma}(G)} \; .
 \end{eqnarray*}
\end{proof}

\section{Compact embedding theorems}
In this section we prove a Rellich-Kondrachov type theorem. As is well-known, this
theorem plays a crucial role in proving compactness of operators and in fixed point arguments.
Now, the standard Rellich-Kondrachov theorem \cite{adams,hebey} is valid only on
spaces with finite measure such as compact Riemannian manifolds. It is then natural to assume that
in our case the condition $\mu_G(G) < \infty$ holds, or equivalently (see \cite{hr1}), that
the locally compact abelian group $G$ is actually compact. We stress that even with this restriction
our results go beyond the standard case: besides infinite products of basic examples of compact abelian
groups, other interesting instances of compact groups are the dyadic group (see for instance \cite{Tat})
and the compact group appearing in the recent preprint \cite{Ta}.

In the theorem below we use the following convention: $g(h) \rightarrow 0$ as $h \rightarrow e$
means that for all $\epsilon > 0$, there exists an open set $U_\epsilon$ with $e \in U_\epsilon$
such that for all $h \in U_\epsilon$ we have $|g(h)| \leq \epsilon$. Also, the notation
$A \hookrightarrow \hookrightarrow B$ means that the space $A$ is compactly embedded into $B$.

\vspace{2mm}

{\bf Remark:} It is known that in the case of ${\mathbb R}^n$, the Kolmogorov-Riesz-Weil theorem
(see \cite{K} and \cite{We}) can be used to prove the Rellich-Kondrachov theorem.
Similar compactness results exist for locally compact abelian groups, see \cite{Fe}, which
presumably would yield another approach to the problem of compact embeddings. We present
a direct proof.

\vspace{2mm}

\begin{thm}\label{comp}
Let  $\frac{1}{\left(1+\gamma(.)^2\right)} \in L^{\alpha}(G^{\wedge})$ for some $\alpha > s$ and
assume that
\begin{equation}\label{cond}
    \frac{|\xi(h)-1|}{(1+\gamma(\xi)^2)^s} \underset{h\rightarrow e}{\longrightarrow} 0 \quad
\text{uniformly with respect to} \,\xi \in G^{\wedge}\, .
\end{equation}
If $G$ is compact, then for all $p< \alpha^*$,
\begin{eqnarray*}
    H^s_{\gamma}(G) \hookrightarrow \hookrightarrow L^p_{\mu_G}(G) \; .
\end{eqnarray*}
\end{thm}

Before proving Theorem \ref{comp} we note
that if $G$ is $\mathbb{R}^n$ or $\mathbb{T}^n$,  Condition (\ref{cond}) is satisfied. Indeed,
if $G=\mathbb{R}^n$, then $G^{\wedge}=\mathbb{R}^n$ and a straightforward calculation
yields
\[
\frac{|\xi(h)-1|}{(1+\gamma(\xi)^2)^s}=\frac{|e^{i \xi h}-1|}{(1+|\xi|^2)^s} \leq |h| \; .
\]
If $G=\mathbb{T}^n$, then  $G^{\wedge}=\mathbb{Z}^n$ and we can show as before that
$\frac{|n(h)-1|}{(1+\gamma(n)^2)^s}=\frac{|e^{i n h}-1|}{(1+|n|^2)^s} \leq |h|$.

\begin{proof}
Let us start with the following lemma.
\begin{lemma}
    Let $f\in H^s_{\gamma}(G)$ and assume that
$\frac{|\xi(h)-1|}{(1+\gamma(\xi)^2)^s} \underset{h\rightarrow e}{\longrightarrow} 0$ uniformly
with respect to $\xi \in G^{\wedge}$. Then, for each $h\in G$
 \begin{eqnarray*}
    \int_G |f(xh)-f(x)|^2d\mu_G(x) \leq C(h) \|f\|^2_{H^s_{\gamma}(G)} \; ,
 \end{eqnarray*}
where $C(h) \underset{h\rightarrow e}{\longrightarrow} 0 $\; .
\end{lemma}
\begin{proof}
By Pontriagin duality we have
\begin{eqnarray*}
    \int_G |f(xh)-f(x)|^2d\mu_G(x) = \int_{G^{\wedge}}|\widehat{f(.h)}(\xi)-\hat{f}(\xi)|^2d\mu_{G^{\wedge}}(\xi) \; .
 \end{eqnarray*}
Since the measure $\mu_G$ is invariant, we obtain
\begin{eqnarray*}
    \widehat{f(.h)}(\xi)=\int_G \overline{\xi(x)} f(xh) d\mu_G(x)=\int_G \overline{\xi(yh^{-1})} f(y) d\mu_G(y)=\\
    \int_G \overline{\xi(y)\xi(h^{-1})} f(y) d\mu_G(y)=\xi(h)\hat{f}(\xi) \; .
\end{eqnarray*}
Hence, we get
\begin{eqnarray*}
    \int_G |f(xh)-f(x)|^2d\mu_G(x) = \int_{G^{\wedge}}|\hat{f}(\xi)|^2 |\xi(h)-1|^2 d\mu_{G^{\wedge}}(\xi)=\\
    \int_{G^{\wedge}}|\hat{f}(\xi)|^2(1+\gamma(\xi)^2)^s \frac{|\xi(h)-1|^2}{(1+\gamma(\xi)^2)^s} d\mu_{G^{\wedge}}(\xi)\leq C(h) \|f\|^2_{H^s_{\gamma}(G)} \; ,
 \end{eqnarray*}
where $C(h)=\|\frac{|\xi(h)-1|^2}{(1+\gamma(\xi)^2)^s}\|_{L^{\infty}(G^{\wedge})}\underset{h\rightarrow e}{\longrightarrow} 0\,$.
\end{proof}
We continue the proof of Theorem \ref{comp}. Let $\mathcal{I}$ be the set of all symmetric
unit-neighborhoods, partially ordered by the inverse inclusion. Then, using Urysohn lemma we can
construct the so-called Dirac net $(\phi_U)_{U\in I}$ in $C_c(G)$ (see \cite{DE}). Each function
$\phi_U$ is nonegative, satisfies $\int_G \phi_U(x) d\mu_G(x)=1$ and the support of $\phi_U$ shrinks.
We are in position to formulate the next lemma:

\begin{lemma}\label{dwa}
    Let $(\phi_U)_{U\in \mathcal{I}}$ be a Dirac net and $f\in H^s_{\gamma}(G)$. Then
    \begin{eqnarray*}
        \int_G |f \ast \phi_U (x)-f(x)|^2d\mu_G(x) \leq  \|f\|^2_{H^s_{\gamma}(G)}\, \sup_{y\in U} C(y) \; .
    \end{eqnarray*}
\end{lemma}
\begin{proof}
\begin{eqnarray*}
        |f \ast \phi_U (x)-f(x)|^2 =\left|\int_G \phi_U(y) f(y^{-1}x)d\mu_G(y) - f(x)\right|^2=\\
        \left|\int_G \phi_U(y)( f(y^{-1}x) - f(x))d\mu_G(y)\right|^2\leq \int_G \phi_U(y)|f(y^{-1}x) - f(x)|^2 d\mu_G(y)=\\
        \int_U \phi_U(y)|f(y^{-1}x) - f(x)|^2 d\mu_G(y)\; .
 \end{eqnarray*}
Hence, by Fubini theorem and the invariance of the measure we get
\begin{eqnarray*}
      \int_G |f \ast \phi_U (x)-f(x)|^2 d\mu_G(x) & \leq & \int_G \int_U \phi_U(y)|f(y^{-1}x) - f(x)|^2 d\mu_G(y) d\mu_G(x) \\
       & = & \int_U \int_G \phi_U(y)|f(y^{-1}x) - f(x)|^2  d\mu_G(x) d\mu_G(y) \\
       & = & \int_U \phi_U(y) \int_G |f(z) - f(yz)|^2  d\mu_G(z) d\mu_G(y) \; .
 \end{eqnarray*}
 By the previous lemma
\begin{eqnarray*}
      \int_G |f \ast \phi_U (x)-f(x)|^2 d\mu_G(x) \leq \int_U \phi_U(y) \|f\|^2_{H^s_{\gamma}(G)} C(y) d\mu_G(y)\leq\\
      \int_U \phi_U(y) d\mu_G(y) \|f\|^2_{H^s_{\gamma}(G)}\sup_{y\in U} C(y) =  \|f\|^2_{H^s_{\gamma}(G)}\sup_{y\in U} C(y) \; .
 \end{eqnarray*}
 This finishes the proof of the lemma.
\end{proof}
Now we can finish the proof of the theorem. Let us take any sequence $f_n$ bounded in the space
$H^s_{\gamma}(G)$, then by Theorem \ref{em} the sequence is bounded in $L^{\alpha^{*}}_{\mu_G}(G)$. Hence,
there exists a subsequence $f_{n_k}$ of $f_n$ such that
\begin{eqnarray*}
      f_{n_k} \rightharpoonup f \quad \text{in} \quad L^{\alpha^{*}}_{\mu_G}(G) \; .
\end{eqnarray*}
We claim that $f_{n_k} \rightarrow f$ in $L^q_{\mu_G}(G)$, where $q< \alpha^*$: for every $f \in L^2_{\mu_G}(G)$ we
denote by $f_{(U)}$ the function $f_{(U)}=f\ast \phi_U$. Also, for simplicity, we write $f_n$ instead of
$f_{n_k}$. By lemma \ref{dwa} we get
 \begin{eqnarray*}
       \sup_{n} \int_G |f_{n_{(U)}}(x)-f_n(x)|^2d\mu_G(x) \leq \sup_{n}  \|f_n\|^2_{H^s_{\gamma}(G)}\sup_{y\in U} C(y)\leq C \sup_{y\in U} C(y)\; .
 \end{eqnarray*}
Moreover, we can show that $\|f_{(U)}-f\|_{L^2_{\mu_G}(G)} \rightarrow 0$ in the sense that for each
$\epsilon>0$ there exists $U_{\epsilon}$ such that for each $U\in \mathcal{I}$, $U \subset U_{\epsilon}$
the inequality holds $\|f_{(U)}-f\|_{L^2_{\mu_G}(G)} \leq \epsilon$. Next, by Minkowski inequality we have
\begin{eqnarray*}
    \|f_n-f\|_{L^2_{\mu_G}(G)}\leq  \|f_n-f_{n_{(U)}}\|_{L^2_{\mu_G}(G)} + \|f_{n_{(U)}}-f_{(U)}\|_{L^2_{\mu_G}(G)} +
\|f_{(U)}-f\|_{L^2_{\mu_G}(G)} \; .
\end{eqnarray*}
Now we fix $\epsilon >0$. There exists $U_{\epsilon} \in I$ such that for each $U\in \mathcal{I}$,
$U \subset U_{\epsilon}$ the following inequality holds
\begin{eqnarray*}
    \|f_n-f\|_{L^2_{\mu_G}(G)}\leq  \frac{2}{3}\epsilon + \|f_{n_{(U)}}-f_{(U)}\|_{L^2_{\mu_G}(G)}.
\end{eqnarray*}
Thus, in order to show that $\|f_n-f\|_{L^2_{\mu_G}(G)} \rightarrow 0$, it is enough to check the limit
$$\|f_{n_{(U_{\epsilon})}}-f_{(U_{\epsilon})}\|_{L^2_{\mu_G}(G)}\underset{n\rightarrow \infty}{\longrightarrow} 0 \; .$$
In fact, since $f_{n} \rightharpoonup f$ in $L^{\alpha^{*}}_{\mu_G}(G)$, we have
\begin{eqnarray*}
    f_{n_{(U_{\epsilon})}} (x)=\int_G \phi_{U_{\epsilon}}(xy^{-1}) f_n(y) d\mu_G(y) \rightarrow \int_G \phi_{U_{\epsilon}}(xy^{-1}) f(y) d\mu_G(y) =f_{(U_{\epsilon})} (x)\; .
\end{eqnarray*}
Moreover, since $G$ is abelian we get
\begin{eqnarray*}
    |f_{n_{(U_{\epsilon})}}-f_{(U_{\epsilon})}|^2 =
\left|\int_G(f_n(y)-f(y))\phi_{U_{\epsilon}}(y^{-1}x)d\mu_G(y)\right|^2 \leq \\
\int_G|f_n(y)-f(y)|\phi_{U_{\epsilon}}(y^{-1}x)d\mu_G(y)\leq \sup_{z \in U_{\epsilon}} \phi_{U_{\epsilon}}(z) \int_G|f_n(y)-f(y)|d\mu_G(y) \; ,
\end{eqnarray*}
and finally, since we are assuming that $G$ is of finite measure, we can apply the Lebesgue theorem and
obtain
\begin{eqnarray*}
    \|f_{n_{(U_{\epsilon})}}-f_{(U_{\epsilon})}\|_{L^2_{\mu_G}(G)}\underset{n\rightarrow \infty}{\longrightarrow} 0 \; .
\end{eqnarray*}
So, we have obtained that $f_{n_k} \rightarrow f$ in $L^2_{\mu_G}(G)$. Finally, since the sequence is
bounded in $L^{\alpha^*}_{\mu_G}(G)$ we can apply Vitali convergence theorem and we obtain that
$f_{n_k} \rightarrow f$ in $L^p_{\mu_G}(G)$, where $p < \alpha^*$.
\end{proof}

\section{An application: the generalized euclidean bosonic string}
We recall that the generalized euclidean bosonic string equation \cite{CMN,GPR_CAOT} is
\begin{equation} \label{abs1}
\Delta e^{-c\Delta}\phi = U(x,\phi)\; , \; \; \; c > 0 \; .
\end{equation}
Classically, Equation (\ref{abs1}) is known as an {\em equation with an infinite number of
derivatives} (see \cite{BK1} and references therein). Such equations have been considered in the
mathematical literature since the 1930's, but only recently physicists have found reasons to study
{\em nonlinear} equations such as (\ref{abs1}), see for instance \cite{BK1,CMN,Ra,V,W}. The existence of
very serious proposals claiming that $p$-adic and non-commutative mathematics are relevant to physics (see
for instance \cite{V,W}), makes it natural, even necessary, to consider equations of physical relevance in
contexts other than Euclidean space or (pseudo-)Riemannian manifolds, as stated in Section 1.

Suppose for a moment that we are working on Euclidean space and that $f$ is the real function
$f(s) = s \, \exp(-c\,s)$ so that, formally, the left hand side of (\ref{abs1}) is $f(\Delta)$.
We expand $f$ as a power series, $f(s) = \sum_{n=0}^\infty \frac{f^{(n)}(0)}{n!}\, s^n$. Then, formally,
we should have
\[
f(\Delta) u = \sum_{n=0}^\infty \frac{f^{(n)}(0)}{n!}\, \Delta^n\,u \; .
\]
Applying Fourier transform we obtain (we set $\widehat{f} = {\mathcal F}(f)$ for clarity)
\begin{eqnarray*}
{\mathcal F}(f(\Delta) u)(\xi) & = & \sum_{n=0}^\infty \frac{f^{(n)}(0)}{n!}{\mathcal F}(\Delta^n\,u) \\
  =  \sum_{n=0}^\infty \frac{f^{(n)}(0)}{n!} (- |\xi|^2)^n {\mathcal F}(u)
 & = & f(- |\xi|^2)\, {\mathcal F}(u)(\xi) \; ,
\end{eqnarray*}
so that, naturally, we may interpret $f(\Delta) u$ in a way that reminds us of the classical definitions of
pseudo-differential operators, see for instance \cite{RT}, as
\begin{equation} \label{pdo}
f(\Delta) u = {\mathcal F}^{-1} \left(\, f(- |\xi|^2)\, {\mathcal F}(u)(\xi) \,\right) = - {\mathcal F}^{-1} \left(\, |\xi|^2 e^{c|\xi|^2}\, {\mathcal F}(u)(\xi) \,\right) \; .
\end{equation}

Motivated by these remarks, we make two definitions. First, we write down the correct domain for the
operator $L_c =  \Delta e^{-c\Delta} - Id$. Then, we define the action of $L_c\,$:

\begin{defn}
The space $\mathcal{H}^{c,\infty}(G)$, $c >0$, is given by
\begin{eqnarray*}
\mathcal{H}^{c,\infty}(G) = \left\{ f \in L^2_{\mu_G}(G) :
\int_{G^{\wedge}} \left( 1 +\gamma(\xi)^2 e^{c \gamma(\xi)^2}
\right)^2 |\widehat{f}(\xi)|^2\, d\mu_{G^{\wedge}}(\xi) < \infty
\right\}.
\end{eqnarray*}
\end{defn}

\begin{defn} \label{def1}
The operator \, $L_c =  \Delta e^{-c\Delta} - Id$ is defined as
\begin{equation} \label{h7}
L_c u = -\mathcal{F}^{-1}\left(\mathcal{F}(u)+\gamma(\xi)^2 e^{c \gamma(\xi)^2} \mathcal{F}(u)\right) \; ,
\end{equation}
for any $u \in \mathcal{H}^{c,\infty}(G)$.
\end{defn}

We note that $L_c $ is an isometry from $\mathcal{H}^{c,\infty}(G)$ into $L^2_{\mu_G}(G)$; we also remark
that analogous definitions of pseudo-differential operators appear in $p$-adic analysis, see for
instance \cite{RV}.

We state two important technical observations on the structure of the space $\mathcal{H}^{c,\infty}(G)$:

\begin{lemma}\label{emb}
\begin{enumerate}
\item For each non-negative $s \in {\mathbb R}$ the embedding
      $\mathcal{H}^{c,\infty}(G) \hookrightarrow H^s_\gamma(G)$ holds.
      In other words,
      $\| f \|_{H^s_\gamma(G)} \leq C(s) \| f \|_{\mathcal{H}^{c,\infty}(G)}$
      for some constant $C(s) > 0$.
\item Assume that $\frac{1}{(1+\gamma(\cdot)^2)^2} \in L^1(G^\wedge)$. Then, the embedding
      $\mathcal{H}^{c,\infty}(G) \hookrightarrow C(G)$ holds.
\end{enumerate}
\end{lemma}
\begin{proof}
The first claim follows immediately from the elementary properties of the map $x \mapsto e^x$.
The second claim is a consequence of the first one combined with our Sobolev embedding result, Theorem 1.
\end{proof}

Now we show that the linear problem $L_cu = g$, $g \in L^2_{\mu_G}(G)$, can be solved completely using
our set-up:

\begin{thm}\label{linear1}
For each $c > 0$ and $g \in L^2_{\mu_G}(G)$, there exists a unique solution
$u_g \in \mathcal{H}^{c,\infty}(G)$ to the linear problem
\begin{eqnarray}\label{problem}
    L_c u = g \; .
\end{eqnarray}
Moreover, the equation
\begin{eqnarray}
   \|u_g \|_{\mathcal{H}^{c,\infty}(G)} =\| g \|_{L^2_{\mu_G}(G)} \label{for1}
\end{eqnarray}
holds.
\end{thm}
\begin{proof}
It is easy to see that $u_g$ given by
\begin{eqnarray*}
   u_g = - \mathcal{F}^{-1}\left(\frac{\mathcal{F}(g)}{1 + \gamma(\xi)^2
   e^{c \gamma(\xi)^2}}\right) \; ,
\end{eqnarray*}
is an element of ${\mathcal H}^{c,\infty}(G)$ which
solves Equation (\ref{problem}). Now, applying Fourier transform
we get
\begin{eqnarray*}
    \left( 1 +\gamma(\xi)^2 e^{c\gamma(\xi)^2} \right)\mathcal{F}(u_g)=
    \mathcal{F}(g)\; ,
\end{eqnarray*}
and so the Plancherel Theorem implies that $\|u_g\|_{\mathcal{H}^{c,\infty}(G)} = \| g \|_{L^2_{\mu_G}(G)}$.
Equation (\ref{for1}) tells us that the operator $L_c$ has trivial kernel.
Uniqueness then follows immediately.
\end{proof}

We are ready to show that the generalized bosonic string equation (\ref{abs1}) admits continuous solutions:

\begin{thm}\label{nonli}
Assume that $G$ is a compact abelian group, and that
$\frac{1}{(1+\gamma(\cdot)^2)} \in L^{\delta}(G^\wedge)$, where $\delta >1$.  Let $U : G \times {\mathbb R}
\rightarrow {\mathbb R}$ be a function which is differentiable
with respect to its second argument, and suppose that there exist
constants $\alpha > 1$, $\beta \in [0, \alpha-1]$, $C > 0$, and  functions $h
\in L^2_{\mu_G}(G)$ and $f\in L^{\frac{2\alpha}{\alpha-1}}_{\mu_G}(G)$, such that the following two inequalities
hold:
\begin{equation} \label{condi1}
    |U(x,y) - y|\leq C(|h(x)| + |y|^{\alpha}), \quad
    \left|\frac{\partial}{\partial y}( U(x,y) - y) \right|\leq C\left(|f(x)|+|y|^{\beta}\right).
\end{equation}
If $\| h \|_{L^2_{\mu_G}(G)}$ is suitably small and $\frac{|\xi(h)-1|}{(1+\gamma(\xi)^2)^{\delta-\frac{\delta}{\alpha}}} \underset{h\rightarrow e}{\longrightarrow} 0$ uniformly
with respect to $\xi \in G^{\wedge}$, then there exists a solution $\phi \in \mathcal{H}^{c,\infty}(G)\cap C(G)$
to the nonlinear problem
\begin{eqnarray}\label{m-jeden3}
    \Delta e^{c\, \Delta}\phi - U(x,\phi)=0 \; .
\end{eqnarray}
\end{thm}
\begin{proof}
Let us  set $V(\cdot , u) = U(\cdot , u) - u$. Then, the nonlinear equation
(\ref{m-jeden3}) is formally equivalent to $L_c u = V(.,u)$, and it is easy to see that the function $V$ belongs to
$L^2_{\mu_G}(G)$. We define the set
\begin{eqnarray*}
    Y_{\epsilon} =\left\{ u \in L^{2\alpha}_{\mu_G}(G) : \|u\|_{L^{2\alpha}_{\mu_G}(G)} \leq \epsilon \right\}
\end{eqnarray*}
for $\epsilon > 0$. It is easy to see that $Y_{\epsilon}$ is a bounded, closed, convex and nonempty subset of
the Banach space $L^{2\alpha}_{\mu_G}(G)$. We define a map $\mathcal{G}$ as follows:
\begin{eqnarray*}
  \mathcal{G}: Y_{\epsilon} \rightarrow L^{2\alpha}_{\mu_G}(G), \quad \mathcal{G}(u) =\tilde{u} \; ,
\end{eqnarray*}
where $\tilde{u}$ is the unique solution to the non-homogeneous linear problem
\begin{eqnarray*}
   L_c \tilde{u} = V(\cdot , u) \; .
\end{eqnarray*}
Lemma 3 and Theorem 3 imply that $\mathcal{G}$ is well defined.
We show that there exists $\epsilon > 0$ such that $\mathcal{G}: Y_{\epsilon} \rightarrow Y_{\epsilon}\,$.
Indeed, let us take $u \in Y_{\epsilon}\,$, then we get, using (\ref{condi1}),
\begin{eqnarray}\label{wazz}
  \|\mathcal{G}(u)\|_{H^{c, \infty}(G)}^2  = \|V(\cdot,u)\|_{L^2_{\mu_G}(G)}^2 \leq
  C^2 \int_G \left| |h(x)| + |u(x)|^\alpha \right|^2 d\mu_G(x) \nonumber \\
   \leq  2 C^2 \int_G \left| |h(x)|^2 + |u(x)|^{2\alpha} \right| d\mu_G(x)
    \leq  2C^2 \left(\|h\|_{L^{2}_{\mu_G}(G)}^2 + \|u\|^{2\alpha}_{L^{2\alpha}_{\mu_G}(G)}\right).  \label{ineq2}
\end{eqnarray}
Now, let us fix $s\in(\delta-\frac{\delta}{\alpha},\delta)$. Using again Lemma \ref{emb} and Theorem \ref{em} we have
\begin{eqnarray*}
  H^{c, \infty}(G)\hookrightarrow H^s_\gamma(G) \hookrightarrow L^{2\alpha}_{\mu_G}(G).
\end{eqnarray*}
Hence, since $u \in Y_\epsilon$, the inequality (\ref{wazz}) there exists a constant $D$ such that
\[
\|\mathcal{G}(u)\|_{L^{2\alpha}(G)}^2 \leq D \left(\|h\|_{L^{2}_{\mu_G}(G)}^2 + \epsilon^{2\alpha} \right).
\]
Since we are assuming that $\|h\|_{L^2_{\mu_G}(G)}$ is suitably small and $\alpha>1$, we can find $\epsilon$ such that
$\|\mathcal{G}(u)\|_{L^{2\alpha}_{\mu_G}(G)} \leq \epsilon$. This implies that
$\mathcal{G}: Y_{\epsilon} \rightarrow Y_{\epsilon}\,$.

Now we apply a fixed point argument. We skip the details, as similar proofs appear in \cite{GPR_JDE}. First,
we note that Theorem 4 implies that $H^{s}_\gamma(G) \hookrightarrow\hookrightarrow L^{2\alpha}_{\mu_G}(G)$, and therefore the map $\mathcal{G}$ is compact. Second, a standard reasoning using the Mean Value Theorem and our assumptions on the
derivative of $V$, implies that the map $\mathcal{F}$ is continuous.
Application of Schauder's fixed point theorem finishes the proof.
\end{proof}

\subsection*{Acknowledgements}
We are most grateful to Professor H. G. Feichtinger for his comments on a first version of
this paper, and for directing us to references \cite{Fe,FG,FeiWe}.
P. G\'orka is partially supported by FONDECYT grant \#3100019;
E.G. Reyes is partially supported by FONDECYT grants \#1070191 and  \#1111042.

\end{document}